\documentclass{ocg}
\usepackage{graphicx,tikz}

\newtheorem{fact}[theorem]{Fact}

\begin{document}

\title{EXTENDED FINITE AUTOMATA AND DECISION PROBLEMS FOR MATRIX SEMIGROUPS} 
\author{\"{O}zlem Salehi}
\thanks{\"{O}zlem Salehi is partially supported by T\"{U}B\.{I}TAK 
	(Scientific and Technological Research Council of Turkey).} 
\author{Ahmet Celal Cem Say}

\address{Bo\v{g}azi\c{c}i University, Department of Computer Engineering, \\
	Bebek 34342 \.{I}stanbul, Turkey\\
	\email{\{ozlem.salehi,say\}@boun.edu.tr}}

%

\maketitle

\begin{abstract}
We make a connection between the subgroup membership and identity problems for matrix groups and extended finite automata. We provide an alternative proof for the decidability of the subgroup membership problem for $ 2 \times 2 $ integer matrices. We show that the emptiness problem for extended finite automata over $ 4 \times 4  $ integer matrix semigroups is undecidable. We prove that the decidability of the universe problem for extended finite automata is a sufficient condition for the decidability of the subgroup membership and identity problems.
\end{abstract}

\section{Introduction}

Among the various extensions of classical finite state automata, extended finite automata over a monoid $ M $ or $ M $-automata have been investigated both implicitly and explicitly by many researchers \cite{Co05,Ka06,MS97}. An $ M $-automaton is a nondeterministic finite automaton equipped with a register that is multiplied by an element of the monoid $ M $ at each step. The register is initialized with the identity element of the monoid and a successful computation is the one which ends in an accept state with the register being equal to the identity element. 

In this paper, our aim is to make a connection between the theory of extended finite automata and the subgroup membership and identity problems for matrix semigroups. Matrices play an important role in various areas of computation, which makes it interesting to study decision problems on matrices. Even for integer matrices of low dimension, many decision problems become non-trivial for finitely generated infinite semigroups.

Let $ S $ be a matrix semigroup finitely generated by a generating set of square matrices $ F $. The \textit{membership problem} is to decide whether or not a given matrix $ Y $ belongs to the matrix semigroup $ S $ \cite{Ma47}. Equivalently, given a finite set of matrices $  F= \{Y_1,Y_2, \dots ,Y_n\} $ and a matrix $ Y $, the problem is to determine if there exists an integer $ k \geq 1 $ and $ i_1, i_2, \dots , i_k \in \{1, \dots , n\} $ such that
$ Y_{i_1} Y_{i_2} \cdots Y_{i_k} = Y $. The identity problem is a special case of the membership problem where $ Y $ is restricted to be the identity matrix. 
 Introduced by Mihailova \cite{Mi68}, the subgroup membership problem is one of the classical decision problems in group theory. Given elements $ h_1,h_2,\dots,h_n $ and $ g $ of a group $ G $, the subgroup membership problem for $ H $ in $ G $ asks whether $ g $ belongs to the subgroup $ H $ generated by $ h_1,h_2,\dots,h_n $. Note that the subgroup membership problem for matrix groups is a special case of the membership problem. 
 
  	For $ 2\times 2 $ integer matrices, the decidability of the membership problem is proven in $ \cite{PS17} $.  For 3x3 matrices, both the identity and membership problems are still open. Undecidability of the membership problem for $ 4 \times 4 $ integer matrices is known for a long time due to a result by Mihailova \cite{Mi68} whereas the undecidability of the identity problem is proven recently in \cite{BP10,KNP17}. 

For our purposes, we define $ S $-automata or extended finite automata over semigroups, generalizing the notion of $ M $-automata from monoids to semigroups. The emptiness problem is defined as the problem of deciding whether a given machine accepts any string. For $ 2 \times 2 $ integer matrices, by using the decidability of the emptiness problem of the corresponding extended finite automata, we provide an alternative proof for the decidability of subgroup membership problem. We show that the undecidability of the identity problem for $ 4 \times 4 $ integer matrices yields the undecidability of the emptiness problem for extended finite automata over semigroups of $ 4 \times 4 $ integer matrices. We also prove some results on the the decidability of the universe problem for extended finite automata, the problem of deciding whether a given machine accepts every string. 
\section{Background}
\subsection{Preliminaries}
 
We denote by $ \mathbb{Z}^{n \times n} $ the set of $ n\times n $ matrices with integer entries. $GL(n,\mathbb{Z})$ denotes the general linear group of degree $ n $
over the ring of integers, equivalently the
group of $ n\times n $ invertible matrices with integer entries. Note
that these matrices have determinant $\pm 1$. Restricting the matrices
in $GL(n,\mathbb{Z})$ to those that have determinant 1, we obtain the
special linear group of degree $ n $ over the ring of integers,
$SL(n,\mathbb{Z})$. We denote the \textit{free group} over $r$ generators by $ \mathbf{F}_r $.

\textit{Word problem} for $ G $ is the subgroup membership problem for the trivial group generated by 1. In other words, given an element $ g \in G $, the problem is to decide whether $ g $ represents the identity element. The word problem language of $ G $ is the language $ W(G,X) $ over $ A = X \cup X^{-1}$ and consists of all words that represent the identity element of $ G $. Most of the time, the statements about word problem are independent of the generating set and in these cases the word problem language is denoted by $ W(G) $.

\subsection{$ S $-automaton}
Let $Q$ be the set of
states, where $q_0 \in Q$ denotes the initial state, $Q_a \subseteq Q$ denotes the
set of accepting states, and let $\Sigma$ be the input alphabet where $ \Sigma_{\varepsilon} =\Sigma \cup \{\varepsilon\} $.

Let $ S $ be a semigroup. An \textit{$ S $-automaton} (extended finite automaton over $ S $) is a
6-tuple
\[ V = (Q, \Sigma,S,\delta, q_0,Q_a) \]
where  the transition function $\delta$ is defined as
\[\delta: Q \times \Sigma_{\varepsilon} \rightarrow \mathbb{P}(Q\times S).\] 
$ \delta(q,\sigma) \ni (q',m) $ means that when $V$ reads the
symbol (or empty string) $\sigma \in \Sigma_{\varepsilon}$
in state $q$, it will move to state $q'$, and write $ x m $ in the register, 
where $ x $ is the old content of the register. 

An $ S $-automaton is in fact an extended finite automaton or a group/monoid automaton \cite{DM00, Ka06} where the group/monoid condition is loosened to a semigroup. In order to define the initialization and acceptance steps, we need an identity element. If $ S $ is a monoid or a group, then an identity element already exists and belongs to $ S $. Otherwise, we define 1 to be the identity element of $ S $. The register of $ V $ is initialized with the identity element $ 1 $ and an input string is accepted if, after completely reading the string, $V$ enters an accept state with the content of the register being equal to the identity element. Note that when S is not a monoid nor a group, then V can accept only the empty string. Nevertheless, we define the concept of S-automaton so that the machines in the proofs of Theorem \ref{theorem: ie} and \ref{theorem: iu} are constructed properly.

We denote by $ L(V) $ the set of accepted strings by $ V $. $ \mathfrak{L}(S) $ denotes the class of languages recognized by $ S $-automata.

Alternatively, monoid automata can be defined through rational subsets as in \cite{Gi96,Ka06} which we discuss next. 

A \textit{finite automaton} $ F $ over a monoid $ M $ is a finite directed graph whose edges are labeled by elements from $ M $. $ F $ consists of a vertex labeled as the initial vertex and a set of vertices labeled as the terminal vertices such that an element of $ M $ is accepted by $ F $ if it is the product of the labels on a path from the initial vertex to a terminal vertex. A subset of $ M $ is called \textit{rational} if its elements are accepted by some finite automaton over $ M $. 

When $ M $ is a free monoid (such as $ \Sigma^* $), then the accepted elements are words over $ \Sigma $ and the set of accepted words is a language over $ \Sigma $. Rational subsets of a free monoid are called rational (regular) languages. Note that when $ M=\Sigma^* $, then the definition coincides with the definition of a finite state automaton.

An $ M $-automaton $ V $ recognizing a language over alphabet $ \Sigma $ is a finite automaton $ F $ over the monoid $ \Sigma^* \times M $ such that the accepted elements are $ (w,1) $ where $ w \in \Sigma^* $. This is stated explicitly in the following proposition by Corson (\cite{Co05}, Proposition 2.2). The proof involves constructing an $ M $-automaton from a finite automaton over $ \Sigma^* \times M $ and vice versa.

\begin{fact}\label{fact: corson}\textup{\cite{Co05}} Let $ L $ be a language over an alphabet $ \Sigma $. Then $ L \in \mathfrak{L}(M) $ if and only if there exists a rational subset $ R \subseteq \Sigma^* \times M $ such that $ L = \{w  \in  \Sigma^∗ | wR1\} $.
\end{fact}

\section{Decidability of the subgroup membership problem for $ \mathbb{Z}^{2 \times 2} $}

It is proven that the membership problem for subsemigroups of $ \mathbb{Z}^{2 \times 2}$ is decidable in \cite{PS17}. In this section, we provide an alternative automata theoretic proof for the decidability of the subgroup membership problem for $ \mathbb{Z}^{2 \times 2}$.

For a finite index subgroup $ H $ of some finitely generated group $ G $, it is known that $ \mathfrak{L}(H)=\mathfrak{L}(G) $ \cite{Co05}. We will go over the proof details and use Fact \ref{fact: corson} to show that given a $ G $-automaton, one can construct an $ H $-automaton recognizing the same language.

\begin{lemma}\label{lem: gh}
	Let $ G $ be a finitely generated group and let $ H $ be a subgroup of finite index. Any $ G $-automaton can be converted into an $ H $-automaton recognizing the same language.
\end{lemma}
\begin{proof}
Let $ X $ be the generator set for $ G $ and let $ A=X \cup X^{-1} $. Let $ V $ be a $ G $-automaton recognizing language $ L $ over alphabet $ \Sigma $. Then there exists a rational subset $ R \subseteq \Sigma^* \times G $ such that $ L=\{w \in \Sigma^* |wR1 \} $. One can define the elements of $ G $ in terms of $ A $ to obtain a rational subset $ R_0 \subseteq \Sigma^* \times A^* $.
	
Since $ H $ has finite index in $ G $, $ W(G) \in \mathfrak{L}(H)$ (\cite{Co05} Lemma 2.4). It follows that there exists a rational subset $ S \subseteq A^* \times H $ such that $ W(G)=\{w \in A^* | wS1\} $. 
	
	 Then the composition $ R_0 \circ S $ is a rational subset of $ \Sigma^* \times H $ and it follows that $ L=\{w \in \Sigma^* | w(R_0 \circ S)1\} $ (\cite{Co05}, Theorem 3.1). The detailed construction of the finite automaton recognizing the composition is given in (\cite{Gi96}, Theorem 5.3). Hence a finite automaton $ F $ over $ \Sigma^* \times H $ recognizing $ L $ exists, from which an $ H $-automaton $ V' $ recognizing $ L $ can be constructed.  
\end{proof}

The following construction of a pushdown automaton simulating an $ \mathbf{F}_2 $-automaton is left as an exercise in \cite{Ka06}. We present here some details of the construction.

\begin{lemma}\label{lemma: f2pda}
	Any $ \mathbf{F}_2 $-automaton can be converted into a pushdown automaton recognizing the same language.  
\end{lemma}
\begin{proof}
	Let $ V $ be an $ \mathbf{F}_2 $-automaton recognizing language $ L $ over $ \Sigma $ with the state set $ Q $ and let $ X=\{a,b\} $ be the generator set for $ \mathbf{F}_2 $. Let us construct a pushdown automaton $ V' $ recognizing the same language with the stack alphabet $ X $. Let $ (q',f) \in \delta (q,\sigma) $ be a transition of $ V $ where $ q,q' \in Q $, $ \sigma \in \Sigma_{\varepsilon} $ and $ f \in \mathbf{F}_2 $ such that $ f=f_1f_2\dots f_n$ where $ f_i \in A=X \cup X^{-1} $ for $ i=1\dots n $. In $ V', $ we need an extra $ n$ states $ q_1\dots q_{n} \notin Q$ to mimic each given transition of $ V $. If $ f_i =a  $ or $ f_i =b $, then this corresponds to pushing $ a $ or $ b $ to the stack, respectively. Similarly, if $ f_i =a^{-1}  $ or $ f_i =b^{-1} $, then $ V' $ pops $ a $ or $ b $ from the stack. Each single transition of $ V $ is accomplished by the pushdown automaton $ V' $ by going through the extra states and pushing and popping symbols. Initially, the register of $ V $ is initialized with the identity element of $ \mathbf{F}_2 $, which corresponds to the stack of $ V' $ being empty. The acceptance condition of $ V $, which is ending in an accept state with the register being equal to the identity element is realized in $ V' $ by starting with an empty stack and accepting with an empty stack in an accept state. We conclude that $ V' $ recognizes language $ L $.   
\end{proof}

\begin{theorem}\label{theorem: me}
	Let $ H $ be a finitely generated subgroup of $ G $. If the emptiness problem for $ G $-automata is decidable, then the subgroup membership problem for $ H $ in $ G $ is decidable.
\end{theorem}
\begin{proof}
	The subgroup membership problem for $ H $ in $ G $ is the problem of deciding whether a given element $ g \in G $ belongs to $ H $. We are going to construct a $ G $-automaton $ V_1 $ and show that $g \in H $ iff $ L(V_1) $ is nonempty.	
	$ V_1 $ has two states: the initial state $ q_1 $ and the accept state $ q_2 $. The transition function of $ V_1 $ is defined as $ \delta(q_1,a)=(q_2,g) $ and $ \delta(q_2,a)=(q_2,h_i) $ for each $ i=1\dots n $ where the set $\{h_1,\dots,h_n\}  $ generates $ H $.

		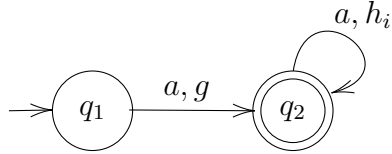
\begin{figure}[h]
		\centering

		\label{fig:me}

	\begin{tikzpicture}[x=0.75pt,y=0.75pt,yscale=-1,xscale=1]
	
	\draw    (190, 110) circle [x radius= 20, y radius= 20]  ;
	\draw    (290, 110) circle [x radius= 20, y radius= 20]  ;
	\draw    (290, 110) circle [x radius= 16, y radius= 16]  ;
	\draw    (208.3,109.6) -- (271,110) ;
	\draw [shift={(271,110)}, rotate = 180.37] [color={rgb, 255:red, 0; green, 0; blue, 0 }  ]   (0,0) .. controls (3.31,-0.3) and (6.95,-1.4) .. (10.93,-3.29)(0,0) .. controls (3.31,0.3) and (6.95,1.4) .. (10.93,3.29)   ;
	
	\draw    (290.15,90) .. controls (297.3,44) and (353.3,90) .. (309.3,101) ;
	\draw [shift={(309.3,101)}, rotate = 343.28999999999996] [color={rgb, 255:red, 0; green, 0; blue, 0 }  ]   (0,0) .. controls (3.31,-0.3) and (6.95,-1.4) .. (10.93,-3.29)(0,0) .. controls (3.31,0.3) and (6.95,1.4) .. (10.93,3.29)   ;
	
	\draw    (148.3,110) -- (170,109.6) ;
	\draw [shift={(170,109.6)}, rotate = 538.94] [color={rgb, 255:red, 0; green, 0; blue, 0 }  ]   (0,0) .. controls (3.31,-0.3) and (6.95,-1.4) .. (10.93,-3.29)(0,0) .. controls (3.31,0.3) and (6.95,1.4) .. (10.93,3.29)   ;

	\draw (190,110) node   {$q_{1}$};
	\draw (290,110) node   {$q_{2}$};
	\draw (237,100) node   {$a ,g$};
	\draw (325,62) node   {$a,h_{i}$};

	\end{tikzpicture}
	
\caption{State transition diagram of $ V_1 $}
	\end{figure}
	
	If $ g \in H $, then it is also true that $ g^{-1} \in H $ since $  H$ is a group. There exists an integer $ k \geq 1 $ and $ i_1, i_2, \dots , i_k \in \{1, \dots , n\} $ such that $ h_{i_1} h_{i_2} \cdots h_{i_k} = g^{-1} $. The string $ a^k $ is accepted by $V_1$ as the register is initially multiplied by $ g$ and there exists a product of elements yielding $ g^{-1} $, from which we can conclude that the identity element can be obtained through a series of transitions of the machine $ V_1 $. Hence, we can conclude that $ L(V_1) $ is nonempty. 
	
	For the other direction, assume that $ L(V_1) $ is nonempty, which means that some input string is accepted by $ V_1 $. Since the acceptance condition requires that the product of the elements multiplied by the register of $ V_1 $ is equal to the identity element and the register is initially multiplied by $ g $, we can conclude that $ H $ contains $ g^{-1}$. Since $  H $ is a group, $ g \in  H $ as well. 
	
	Now suppose that the emptiness problem for $ G $-automaton is decidable. Then one can check if $ g $ is an element of $ H $ by constructing $ V_1 $ and checking if $ L(V_1) $ is nonempty. Hence, the subgroup membership problem for $ H $ is also decidable.   	
\end{proof}

\begin{theorem}
Given a matrix $ Y $ from $ \mathbb{Z}^{2 \times 2} $ and a subgroup $ H$ of $ \mathbb{Z}^{2 \times 2} $, it is decidable whether $ Y $ belongs to $ H $.	
\end{theorem}
\begin{proof}
We are going to show that the emptiness problem for $ \mathbb{Z}^{2 \times 2 } $-automaton is decidable and use Theorem \ref{theorem: me} to conclude the result. 

Suppose that a $ \mathbb{Z}^{2 \times 2 } $-automaton $ V $ is given. When $ V $ processes an input string, its register is initialized by the identity matrix and multiplied by matrices from $ \mathbb{Z}^{2\times 2} $. Suppose that in a successful computation leading to acceptance, the register is multiplied by some non-invertible matrix $ Y $. Since $ Y $ is non-invertible, the register can not be equal to the identity matrix again and such a computation can not be successful. Any such edges labeled by a non-invertible matrix can be removed from $ V $, without changing the accepted language. We can conclude that the matrices multiplied by the register are invertible and belong to $ GL(2,\mathbb{Z}) $ and $ V $ is in fact a $ GL(2,\mathbb{Z}) $-automaton.

Since $ \mathbf{F}_2 $ has finite index in $ GL(2,\mathbb{Z}) $, one can construct an $ \mathbf{F}_2 $-automaton recognizing $ L(V) $ by Lemma $ \ref{lem: gh} $. The $ \mathbf{F}_2 $-automaton can be converted to a pushdown automaton $ V' $ using the procedure described in Lemma \ref{lemma: f2pda}. Since the emptiness problem for pushdown automata is known to be decidable, we conclude that the emptiness problem for $ \mathbb{Z}^{2 \times 2 } $-automata is also decidable since a $ \mathbb{Z}^{2 \times 2 } $-automaton can be converted to a pushdown automaton. Then by Theorem \ref{theorem: me}, the result follows.
\end{proof}

\section{Undecidability of the emptiness problem for $ \mathbb{Z}^{4 \times 4}$-automata}

In \cite{KNP17}, it is proven that the identity problem is undecidable for a semigroup generated by eight $ 4\times 4 $ integer matrices. Using this fact, we prove that the emptiness problem is undecidable for the corresponding semigroup automaton. 
 
\begin{theorem}\label{theorem: ie}
	Let $ S $ be a finitely generated semigroup. If the emptiness problem for $ S $-automata is decidable, then the identity problem for $ S $ is decidable.
\end{theorem}
\begin{proof}
	We are going to construct an $ S $-automaton $ V_2 $ and show that $ S $ contains the identity element iff $ L(V_2) $ is nonempty. $ V_2$ has two states: the initial state $ q_1 $ and the accept state $ q_2 $. The transition function of $ V_2$ is defined as $ \delta(q_1,a)=(q_2,s_i) $ and $ \delta(q_2,a)=(q_2,s_i) $  for each $ i=1\dots n $ where $ \{s_1,s_2,\dots s_n \} $ is the generator set for $ S $. 
	
	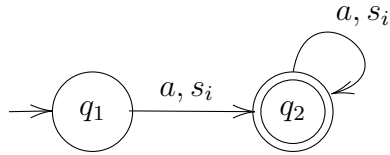
\begin{figure}[h]
		\centering
	\begin{tikzpicture}[x=0.75pt,y=0.75pt,yscale=-1,xscale=1]
	
	\draw    (190, 110) circle [x radius= 20, y radius= 20]  ;
	\draw    (290, 110) circle [x radius= 20, y radius= 20]  ;
	\draw    (290, 110) circle [x radius= 16, y radius= 16]  ;
	\draw    (208.3,109.6) -- (271,110) ;
	\draw [shift={(271,110)}, rotate = 180.37] [color={rgb, 255:red, 0; green, 0; blue, 0 }  ]   (0,0) .. controls (3.31,-0.3) and (6.95,-1.4) .. (10.93,-3.29)(0,0) .. controls (3.31,0.3) and (6.95,1.4) .. (10.93,3.29)   ;
	
	\draw    (290.15,90) .. controls (297.3,44) and (353.3,90) .. (309.3,101) ;
	\draw [shift={(309.3,101)}, rotate = 343.28999999999996] [color={rgb, 255:red, 0; green, 0; blue, 0 }  ]   (0,0) .. controls (3.31,-0.3) and (6.95,-1.4) .. (10.93,-3.29)(0,0) .. controls (3.31,0.3) and (6.95,1.4) .. (10.93,3.29)   ;
	
	\draw    (148.3,110) -- (170,109.6) ;
	\draw [shift={(170,109.6)}, rotate = 538.94] [color={rgb, 255:red, 0; green, 0; blue, 0 }  ]   (0,0) .. controls (3.31,-0.3) and (6.95,-1.4) .. (10.93,-3.29)(0,0) .. controls (3.31,0.3) and (6.95,1.4) .. (10.93,3.29)   ;

	\draw (190,110) node   {$q_{1}$};
	\draw (290,110) node   {$q_{2}$};
	\draw (237,100) node   {$a ,s_i$};
	\draw (325,62) node   {$a,s_i$};

	\end{tikzpicture}
		\caption{State transition diagram of $ V_2 $}
		\label{fig:ie}
	\end{figure}
	
	If $ S $ contains the identity element, then there exists an integer $ k \geq 1 $ and $ i_1, i_2, \dots , i_k \in \{1, \dots , n\} $ such that $ s_{i_1} s_{i_2} \cdots s_{i_k} = 1 $. Then the string $ a^k $ is accepted by $ V_2$ as there exists a product of elements yielding the identity element and this product can be obtained by a series of transitions. Hence, we can conclude that $ L(V_2) $ is nonempty. For the converse, suppose that $ L(V_2)$ is nonempty, which means that some input string is accepted by $ V_2 $. Since the acceptance condition requires that the product of the elements multiplied by the register of $ V_2 $ is equal to the identity element, we can conclude that $ S $ contains the identity element.    
	
	Now suppose that the emptiness problem for $ S $-automaton is decidable. Then one can check if $ S $ contains the identity element by constructing $ V_2 $ and checking if $ L(V_2) $ is nonempty. Hence, the identity problem for $ S $ is also decidable.   	
\end{proof}

The identity problem for $ \mathbb{Z}^{4 \times 4} $ is shown to be undecidable for a semigroup of 48 matrices in \cite{BP10}. Later on, the result is improved to eight matrices in \cite{KNP17}.

\begin{fact}\label{fact: 4id}\textup{\cite{KNP17}}
	Given a semigroup $ S $ generated by eight $ 4\times 4 $ integer matrices, determining whether
	the identity matrix belongs to $ S $ is undecidable.
\end{fact}

Using this result, we obtain the following corollary about the emptiness problem for extended finite automata over semigroups of $ \mathbb{Z}^{4 \times 4} $.

\begin{corollary}\label{corollary: sl4z}
	Let $ S $ be a subsemigroup of $ \mathbb{Z}^{4 \times 4}$ generated by eight matrices. The emptiness problem for $ S $-automaton is undecidable.
\end{corollary}
\begin{proof}
	By Fact $ \ref{fact: 4id}$ we know that the identity problem for $ S $ is undecidable. By Theorem \ref{theorem: ie}, the result follows. 	
\end{proof}

\section{Universe problem for $ S $-automata}

In this section we prove some results connecting the universe problem for $ S $-automata and the subgroup membership and identity problems for $ S $.

\begin{theorem}\label{theorem: mu}
	
	Let $ H $ be a finitely generated subgroup of $ G $. If the universe problem for $ G $-automata is decidable, then the subgroup membership problem for $ H $ in $ G $ is decidable.
\end{theorem}
\begin{proof}
We are going to construct a $ G $-automaton $ V_2 $ and such that $g \in H $ iff $ L(V_3)=\Sigma^* $ where $ \Sigma=\{a\} $. $\{h_1,h_2,\dots h_n \} $ is the generator set for $ H $. 
 
	\begin{figure}[h]
		\centering
		\begin{tikzpicture}[x=0.75pt,y=0.75pt,yscale=-1,xscale=1]
		
		\draw    (190, 110) circle [x radius= 20, y radius= 20]  ;
		\draw    (290, 110) circle [x radius= 20, y radius= 20]  ;
		\draw    (290, 110) circle [x radius= 16, y radius= 16]  ;
		\draw    (190, 110) circle [x radius= 16, y radius= 16]  ;
		\draw    (208.3,109.6) -- (271,110) ;
		\draw [shift={(271,110)}, rotate = 180.37] [color={rgb, 255:red, 0; green, 0; blue, 0 }  ]   (0,0) .. controls (3.31,-0.3) and (6.95,-1.4) .. (10.93,-3.29)(0,0) .. controls (3.31,0.3) and (6.95,1.4) .. (10.93,3.29)   ;
		
		\draw    (290.15,90) .. controls (297.3,44) and (353.3,90) .. (309.3,101) ;
		\draw [shift={(309.3,101)}, rotate = 343.28999999999996] [color={rgb, 255:red, 0; green, 0; blue, 0 }  ]   (0,0) .. controls (3.31,-0.3) and (6.95,-1.4) .. (10.93,-3.29)(0,0) .. controls (3.31,0.3) and (6.95,1.4) .. (10.93,3.29)   ;
		
		\draw    (148.3,110) -- (170,109.6) ;
		\draw [shift={(170,109.6)}, rotate = 538.94] [color={rgb, 255:red, 0; green, 0; blue, 0 }  ]   (0,0) .. controls (3.31,-0.3) and (6.95,-1.4) .. (10.93,-3.29)(0,0) .. controls (3.31,0.3) and (6.95,1.4) .. (10.93,3.29)   ;

		\draw (190,110) node   {$q_{1}$};
		\draw (290,110) node   {$q_{2}$};
		\draw (237,100) node   {$a ,g$};
		\draw (325,50) node   {$a,h_{i}$};
		\draw (325,65) node   {$\varepsilon,h_{i}$};

		\end{tikzpicture}
		\caption{State transition diagram of $ V_3 $}
		\label{fig:mu}
	\end{figure}
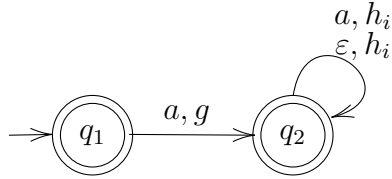
	
	The rest of the proof is similar to the proof of Theorem \ref{theorem: me} and omitted here.
	\end{proof}

\begin{theorem}\label{theorem: iu}
	Let $ S $ be a finitely generated semigroup. If the universe problem for $ S $-automata is decidable, then the identity problem for $ S $ is decidable. 
\end{theorem}
\begin{proof}
	We are going to construct an $ S $-automaton $ V_4 $ such that $ S $ contains the identity element iff $ L(V_4)=\Sigma^* $ where $ \Sigma=\{a\} $. $ \{s_1,s_2,\dots s_n \} $ is the generator set for $ S $. 
	
	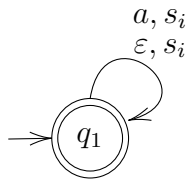
\begin{figure}[h]
		\centering

		\begin{tikzpicture}[x=0.75pt,y=0.75pt,yscale=-1,xscale=1]
		
		\draw    (290.15, 110) circle [x radius= 19.15, y radius= 20]  ;
		\draw    (290.14, 110) circle [x radius= 16.16, y radius= 16.5]  ;
		\draw    (290.15,90) .. controls (297.3,44) and (353.3,90) .. (309.3,101) ;
		\draw [shift={(309.3,101)}, rotate = 343.28999999999996] [color={rgb, 255:red, 0; green, 0; blue, 0 }  ]   (0,0) .. controls (3.31,-0.3) and (6.95,-1.4) .. (10.93,-3.29)(0,0) .. controls (3.31,0.3) and (6.95,1.4) .. (10.93,3.29)   ;
		
		\draw    (249.3,110.4) -- (271,110) ;
		\draw [shift={(271,110)}, rotate = 538.94] [color={rgb, 255:red, 0; green, 0; blue, 0 }  ]   (0,0) .. controls (3.31,-0.3) and (6.95,-1.4) .. (10.93,-3.29)(0,0) .. controls (3.31,0.3) and (6.95,1.4) .. (10.93,3.29)   ;
	\draw (290,110) node   {$q_{1}$};
\draw (325,50) node   {$a,s_{i}$};
\draw (325,65) node   {$\varepsilon,s_{i}$};

		\end{tikzpicture}
		\caption{State transition diagram of $ V_4 $}
		\label{fig:iu}
	\end{figure}
The rest of the proof is similar to the proof of Theorem \ref{theorem: ie} and omitted here.
\end{proof}

Let us note that the converses of Theorem \ref{theorem: mu} and \ref{theorem: iu} are not true. For a given pushdown automaton, an $\mathbf{F}_2 $-automaton recognizing the same language can be constructed \cite{Ka06}. It is a well known fact that the universe problem for pushdown automata is undecidable from which we can conclude that the universe problem for $\mathbf{F}_2 $-automaton is undecidable. On the other hand, $ \mathbf{F}_2 $ is a subgroup of $ SL(2,\mathbb{Z}) $ and the membership problem for $ SL(2,\mathbb{Z}) $ and thus the identity problem are known to be decidable \cite{KNP17}.

 \section{Future work}
 
%

It is still not known whether the membership and identity problems are decidable for semigroups of $ 3 \times 3 $ integer matrices. Recent results from $ \cite{KNP17} $ suggest that these problems are more likely to be decidable. We propose that investigating the decidability of the emptiness and universe problems for extended finite automata defined over $3 \times  3$ integer matrices is one possible way for obtaining results about the decision problems on these matrix semigroups.
 
\bibliographystyle{ocg}
\bibliography{references}

\EndOfArticle 

\end{document}